%% file: conference_101719.tex
\documentclass[9pt,conference]{IEEEtran}
\IEEEoverridecommandlockouts
\usepackage{cite}
\usepackage{amsmath,amssymb,amsfonts}
\usepackage{graphicx}
\usepackage{textcomp}
\usepackage{xcolor}
\def\BibTeX{{\rm B\kern-.05em{\sc i\kern-.025em b}\kern-.08em
    T\kern-.1667em\lower.7ex\hbox{E}\kern-.125emX}}
\usepackage{booktabs}
\usepackage{multicol}
\usepackage{lipsum}

\usepackage[inline]{enumitem}
\usepackage[switch]{lineno}
\usepackage[noend]{algpseudocode}
\usepackage{algorithm}
\usepackage{comment}
\usepackage{upgreek}
\usepackage{url}
\input{symbols.tex}

\begin{document}

\title{ Moving Target Detection via Multi-IRS-Aided OFDM Radar}
\author{Zahra Esmaeilbeig, Arian Eamaz, Kumar Vijay Mishra and Mojtaba Soltanalian
\thanks{Zahra Esmaeilbeig, Arian Eamaz and Mojtaba Soltanalian are with the ECE Departement, University of Illinois at Chicago, Chicago, IL 60607 USA. Email: \{zesmae2, aeamaz2, msol\}@uic.edu.}
\thanks{Kumar Vijay Mishra is with the United States DEVCOM Army Research Laboratory, Adelphi, MD 20783 USA. E-mail: kvm@ieee.org.}
\thanks{This work was sponsored in part by the National Science Foundation Grant  ECCS-1809225, and in part by the Army Research Office, accomplished under Grant Number W911NF-22-1-0263. The views and conclusions contained in this document are those of the authors and should not be interpreted as representing the official policies, either expressed or implied, of the Army Research Office or the U.S. Government. The U.S. Government is authorized to reproduce and distribute reprints for
Government purposes notwithstanding any copyright notation herein.}}
\maketitle
\begin{abstract}
 An intelligent reflecting surface (IRS) consists of passive reflective elements capable of altering impinging waveforms. The IRS-aided radar systems have recently been shown to improve detection and estimation performance by exploiting the target information collected via non-line-of-sight paths. However, the waveform design problem for an IRS-aided radar has remained 
 relatively unexplored. In this paper, we consider a multi-IRS-aided orthogonal frequency-division multiplexing (OFDM) radar and study the theoretically achievable accuracy of target detection. In addition, we jointly design the OFDM signal and IRS phase-shifts to optimize the target detection performance via an alternating optimization approach. To this end, we formulate the IRS phase-shift design problem as a unimodular \emph{bi-quadratic} program which is tackled by a computationally cost-effective approach based on power-method-like iterations. Numerical experiments illustrate that our proposed joint design of IRS phase-shifts and the OFDM code improves the detection performance in comparison with  conventional OFDM radar. 
 
\end{abstract}

\setlength{\abovedisplayskip}{3pt}
\setlength{\belowdisplayskip}{3pt}

\section{Introduction}
 Intelligent reflective surface (IRS) is an  emerging technological advancement  for next-generation wireless systems. An IRS comprises meta-material units that enable smart and programmable radio environments by introducing predetermined phase-shifts to the impinging signal~\cite{bjornson2022reconfigurable}. The IRS-aided wireless communications are shown to provide  range extension to users with obstructed direct links \cite{wu2019towards}, 
 enhance physical layer security~\cite{mishra2022optm3sec,zhang2022security}, facilitate unmanned air vehicle (UAV) communications~\cite{Giovanni2021}, and shaping the wireless channel through multi-beam design 
\cite{torkzaban2021shaping}. Recent works have also introduced IRS to integrated communications and sensing systems~\cite{wang2022stars,wei2022irs,mishra2022optm3sec,elbir2022rise}. 

Recently, following the advances in~\cite{esmaeilbeig2022irs,wang2022stars}, IRS-aided sensing for non-line-of-sight (NLoS) target estimation has been investigated in ~\cite{esmaeilbeig2022irs,esmaeilbeig2022cramer,esmaeilbeig2022joint}.
In~\cite{lu2021intelligent}, the phase-shift matrix of the IRS was optimized for collocated MIMO radar to improve the estimation and detection performance. Target detection was also considered in cases where the radar is aided by a single IRS~\cite{buzzi2021radar} or multiple IRS platforms~\cite{buzzi2022foundation,lu2021target}. The deployment of multiple IRS platforms 
is necessary to overcome line-of-sight (LoS) blockage or obstruction in cases where the NLoS path formed by a single IRS is unable to provide the desired coverage. 
To this end, \cite{wang2021joint} jointly designed the radar transmitter and IRS beamformers for a multi-IRS-aided radar. Similar to a conventional radar \cite{radarsignaldesign2022}, a judicious design of transmit waveforms improves the performance of IRS-aided radar. 

In general, radar waveform design is a well-investigated problem~\cite{radarsignaldesign2022,hu2016locating,soltanalian2013joint}. However, it is relatively unexamined for IRS-aided scenarios. Among prior works, \cite{liu2022joint} designed a transmit radar code with constant-modulus for a narrowband IRS-aided radar. However, wideband signaling compensates for signal fading resulting from multipath propagation \cite{sen2010adaptive}. Therefore, very recent works~\cite{wei2022irs,wei2022multi,wei2022simultaneous} investigate wideband waveforms such as orthogonal frequency-division multiplexing (OFDM) signaling to improve detection with IRS-aided radar. 

In this paper, we focus on designing a wideband radar waveform for multi-IRS-aided radar jointly with the IRS phase-shifts.  
In particular, we formulate the detection problem as a hypothesis test to decide the presence of a target in a particular range cell. Then, we jointly design the OFDM signal and the IRS phase-shifts to enhance the receiver operating characteristics (RoC) associated with moving target detection. We adopt noncentrality parameter of the asymptotic distribution of the generalized likelihood ratio test (GLRT) statistic~\cite{sen2010adaptive} as the performance metric for target detection. We demonstrate that  maximizing the noncentrality parameter with respect to the system parameters such as the transmit waveform and phase-shifts of IRS, yields improvement in the probability of detection. Contrary to prior works, wherein only IRS phase-shifts were optimized in an IRS-aided radar~\cite{esmaeilbeig2022irs,lu2021target,buzzi2022foundation}, we show that jointly optimal waveform and phase-shifts increase the probability of detection. Further, our IRS-aided radar outperforms the multipath OFDM radar \cite{sen2010adaptive}  with specular reflection in the exactly identical paths between the target and radar. 


The remainder of this paper is organized as follows. In the next section, we describe the signal model for the multi-IRS-aided OFDM radar. The moving target detector based on GLRT is introduced in Section~\ref{sec:3}. 
We present our joint waveform and IRS phase-shift design in Section~\ref{sec:5}. We validate our model and methods via numerical experiments in Section~\ref{sec:6} and conclude in Section~\ref{sec:7}.  

Throughout this paper, we use bold lowercase and bold uppercase letters for vectors and matrices, respectively. The $(m,n)$-th element of the matrix $\mbB$ is $\left[\mbB\right]_{mn}$. The sets of complex and real numbers are $\mathbb{C}$ and $\mathbb{R}$, respectively;  $(\cdot)^{\top}$, $(\cdot)^{\ast}$ and $(\cdot)^{\mathrm{H}}$ are the vector/matrix transpose, conjugate, and Hermitian transpose, respectively; the trace of a matrix is  $\operatorname{Tr}(.)$; the function $\textrm{diag}(.)$ returns the diagonal elements of the input matrix; and $\textrm{Diag}(.)$  
produces a diagonal/block-diagonal matrix with the same diagonal entries/blocks as its vector/matrices argument. 
The Hadamard (element-wise) and Kronecker products are $\odot$ and $\otimes$, respectively. The vectorized form of a matrix $\mbB$ is written as $\vec{\mbB}$ and the block diagonal vectorization \cite{xu2006block} is denoted by $\operatorname{vecb}(\mbB)$. The $s$-dimensional all-ones vector and the identity matrix of  size $s\times s$ are $\mathbf{1}_{s}$, and $\mbI_s$, respectively. The  minimum eigenvalue of $\mbB$ is  denoted by $\lambda_{min}(\mbB)$. The real, imaginary, and angle/phase components of a complex number are $\Re{\cdot}$, $\Im{\cdot}$, and $\arg{\cdot}$, respectively.  $\mathrm{vec}_{_{K,L}}^{-1}\left(\mbb\right)$ reshapes the input vector $\mbb\in\mathbb{C}^{KL}$ into a  matrix $\mbB\in\mathbb{C}^{K \times L}$ such that $\vec{\mbB}=\mbb$. Also, $\mathbf{0}_{N}$ is the all-zero vector of size $N$. 
The generalized inversion of a matrix $\mbB$ is $\left(\mbB\right)^{-}$. We use $\textrm{Pr}\left(.\right)$ to denote the probability. 
\section{System Model}
\label{sec:2}
Consider a multi-IRS-aided radar with transmitter and receiver  located at $\uprho_{_r}=[0,0]^T$ in the two-dimensional (2-D) Cartesian coordinate system. The radar transmits an OFDM signal with the bandwidth $B$ Hz consisting of  $L$ subcarriers 
as 
\begin{equation}\label{eq:modell}
s_{_{\textrm{OFDM}}}(t)=\sum_{l=0}^{L-1}a_l e^{\j2\pi f_l t},  \quad  0\leq t \leq T,
\end{equation}
where $a_l$ is the waveform code, $f_l=f_c+l\Delta_f$ denotes the $l$-th subcarrier  frequency and the subcarrier spacing is chosen as $\Delta_f=B/(L+1)=1/T$ to guarantee the orthogonality of the subcarriers. The vector $\ba=[a_1,\ldots,a_L]^{\T}$ collects the OFDM coefficients of all subcarriers for which we have  $\|\ba\|_2^2=1$. The pulses in~\eqref{eq:modell} are transmitted with the pulse repetition interval (PRI) $T_{\PRI}$.

The $M$ IRS platforms denoted as IRS$_{1}$, IRS$_{2}$,..., IRS$_{M}$ are  installed at stationary known locations (Fig.~\ref{fig::4}). Each IRS is a uniform linear array (ULA) with $N_m$ reflecting elements, and with an inter-element spacing of $d$.  The first element of IRS$_{m}$ is located at a known coordinate $\uprho_{_i}^{(m)}=[x_i^{(m)},y_i^{(m)}]^T$. The space-frequency steering vector of the $m$-th IRS is
$\mbb_m(\btheta,f_l)= \left[1,e^{\textrm{j}\frac{2\pi f_l}{c}d sin \theta},\ldots,e^{\textrm{j}\frac{2\pi f_l}{c}d(N_m-1) sin \theta}\right]^{\top}$, where $c$ is the speed of light, $f_l$ is the subcarrier 
frequency, $d$ is the half-wavelength Nyquist spacing and $\theta$ is the direction of the impinging wavefront at the  ULA. Each reflecting element of IRS$_m$ reflects the incident signal with a phase-shift change that is configured via a smart controller~\cite{bjornson2022reconfigurable}.   We denote the phase-shift vector  of IRS$_m$ by $\mbv_m=[e^{\textrm{j}\phi_{_{m,1}}},\ldots,e^{\textrm{j}\phi_{_{m,N_m}}}]^\T \in \mathbb{C}^{N_m}$, 
where $\phi_{_{m,k}}\in[0,2\pi)$ is the phase-shift associated with the $k$-th passive element of  IRS$_m$. Clearly, $\mbv_m$ is   a unimodular vector chosen from the set $\Omega^{N_m}$, where $\Omega = \left\{s \in \mathbb{C}| \mbs= e^{\textrm{j}\omega_{i}}, \omega_i \in [0,2\pi)\right\}$.


Assume a target at $\uprho_{_t}=[x_t,y_t]^{\T}$ moving with velocity $\upnu_t=[\upnu_x,\upnu_y]^{\T}$. In the LoS path, the target is characterized by its Doppler shift and time delay given by, respectively, 
\begin{align}\label{eq::doppler-LOS}
\nu_{0}&=\frac{1}{c}\frac{2 \upnu_t^{\T}(\uprho_{_t}-\uprho_{_r})}{d_{tr}},
\end{align}
and
\begin{align}
\tau_{0}&=\frac{2d_{tr}}{c},
\end{align}
where $d_{tr}=\|\uprho_{_{t}}-\uprho_{_{r}}\|_2$ is the distance between the radar and target. The IRS deployment also yields $M$ non-line-of-sight (NLoS) paths from the  target to the radar. The Doppler shift and time delay in the  radar-IRS$_m$-target-IRS$_m$-radar path are, respectively,
\begin{align}\label{eq::doppler-NLOS}
\nu_{m}&=\frac{1}{c}\frac{2\upnu_t^{\T}(\uprho_{_t}-\uprho^{(m)}_{_{i}})}{d^{(m)}_{it}},
\end{align}
and
\begin{align}
\tau_{m}&=2\frac{d^{(m)}_{ri}+d^{(m)}_{it}}{c},  
\end{align}
for $m=1,\ldots,M$, where 
$d^{(m)}_{ri}=\|\uprho^{(m)}_{_{i}}-\uprho_{_{r}}\|_2$ and
$d^{(m)}_{it}=\|\uprho_{_t}-\uprho^{(m)}_{_{i}}\|_2$
are the radar-IRS$_m$ and target-IRS$_m$ distances, respectively. 


We make the following assumptions about the IRS-aided OFDM radar and target parameters: 
\begin{description}
    \item[A1] ``Bandwidth-invariant Doppler'': The bandwidth of OFDM signal is much smaller than the carrier frequency, i.e., $B \ll f_c$.  Hence, the phase-shifts arising from the Doppler effect are identical over all subcarriers.

    \item[A2] ``Slow Target'': The Doppler frequency of the target does not change during one coherent processing interval (CPI) i.e.  $\nu_m << \frac{1}{NT_{_{\PRI}}B}$. Therefore, the  following piecewise-constant approximation holds $\nu_m t \approx \nu_m n T_{_{\PRI}} $ , for $t \in [nT_{_{\PRI}},(n+1)T_{_{\PRI}}]$. 

    \item[A3] ``Narrow surveillance area'': The radar is deployed in a region, where the range of the target is much greater than the width or cross-range extent of the surveillance area. The relative time gaps between any two signals received from NLoS paths are very small in comparison to the LoS round trip delays, i.e., $\tau_{m} \approx \tau_0 =\frac{2 d_{tr}}{c}$  for $m \in \{1,\ldots,M\}$. 

    \item[A4] ``Frequency-invariant IRS phase-shift'': The IRS platforms impose the same phaseshifts over  all subcarrier frequencies  and  therefore the IRS phase-shift matrix is not indexed over different frequencies, i.e., $\bPhi_m(f_l)=\bPhi_m$, for $l \in \{0,\ldots,L-1\}$ and $m \in \{1,\ldots,M\}$.
    
    
    \item [A5] ``Inter-IRS interference'': The mutual interference between various IRS platforms is negligible. In other words, the interference  caused by reflections in the radar-IRS$_m$-target-IRS$_{m'}$-radar path for $m \neq m'$ is insignificant because IRS is a passive reflector and the reflections in non-beamformed directions are weaker.
\end{description}
\begin{figure}[t]
	\centering
    \includegraphics[width=1\columnwidth]{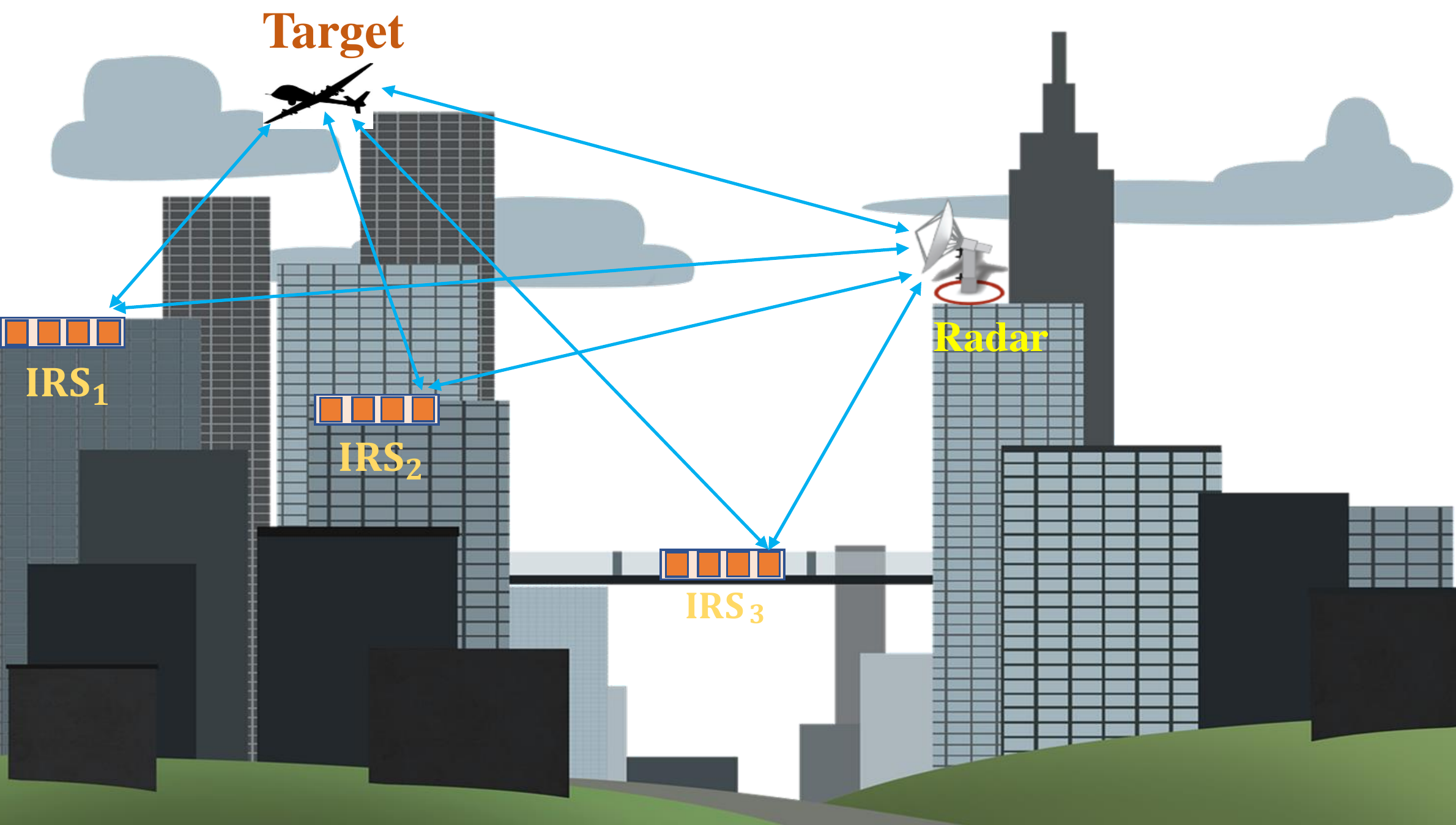}
	\caption{\textcolor{black}{A simplified illustration of various NLoS or virtual LoS links provided by multiple IRS platforms mounted on urban infrastructure between the radar and the hidden moving target.}
	}
	\label{fig::4}
\end{figure}



Define the NLoS  channel  along the  $l$-th subcarrier and $m$-th path as\par\noindent\small
\begin{equation}\label{eq:1}
h_{_{lm}}=\mbb(\theta_{ir,m},f_l)^{\T}\bPhi_m \mbb(\theta_{ti,m},f_l)\mbb(\theta_{it,m},f_l)^{\T}\bPhi_m \mbb(\theta_{ri,m},f_l),    
\end{equation}\normalsize
for  $m>1$ and  $h_{l0}$ is the  LoS channel~\cite{sen2010adaptive}. We  define  $\bPhi_m=\Diag{\mbv_m}$ as diagonalization of the 1-D phase-shift vector of IRS$_m$. 
Assume the LoS path between  radar and target is obstructed, i.e., $h_{l0} \approx 0$. The signal reflected from a Swerling-0~\cite{skolnik2008radar} target with   $\alpha_{_{lm}}$ as the  complex  reflectivity/amplitude along the  $l$-th subchannel and $m$-th path is a delayed, modulated and scaled version of the transmit signal in~\eqref{eq:modell} as \par\noindent\small
\begin{align}\label{eq:y_l}
y_{l}(t)&=\sum_{m=0}^{M} a_l h_{_{l m}} \alpha_{lm} e^{\j2\pi l \Delta_f (1+\nu_m)(t-\tau_m)}\nonumber\\
&\times e^{-\j 2\pi f_c (1+\nu_m)\tau_m} e^{\j2\pi f_c \nu_m t} e^{\j 2\pi f_c t}+w_{l}(t),  
\end{align}\normalsize 
where 
the signal independent interference (noise) for the  $l$-th subcarrier is denoted by $w_{l}(t)$. 


We collect $N$ samples from the signal, at $t=\tau_0+n T_{_{\PRI}}$, $n=0,\ldots,N-1$. By applying  $\nu_m t \approx \nu_m n T_{_{\PRI}} $ (\textbf{A2}) and  $\tau_{m} \approx \tau_0$ (\textbf{A3}) to~\eqref{eq:y_l}, the discrete-time received signal corresponding to  the  range-cell of interest  is
\begin{align}
y_l[n]&= \sum_{m=0}^{M} a_l h_{_{l m}} \alpha_{_{l m}} p_{_{l }}(n,\nu_m)+w_{l}[n],
\end{align}\normalsize
where, $p_{_{l}}(n,\nu_m)=e^{-\j2\pi f_l\tau_0} e^{\j 2\pi f_l \nu_{m}n T_{_{\PRI}}}$ contains the unknown target delay and Doppler.

We stack measurements from all L subchannels to obtain the $L \times 1$ vector 
\begin{equation}
\mby[n]=\Diag{\ba}\mbX\mbp(n,\bnu)+\mbw[n],     
\end{equation}\normalsize
where the 
Doppler steering vector is   $\mbp(n,\bnu)=[\mbp_{_{0}}(n,\bnu)^{\T},\ldots , \mbp_{_{L-1}}(n,\bnu)^{\T}]^{\T}$, with $\mbp_{_{l}}(n,\bnu)=[p_{_{l}}(n,\nu_{_0}),\ldots,p_{_{l}}(n,\nu_{_M})]^{\T}$, and  the  $L \times 1$ noise  vectors  is $\mbw[n]=[w_{0}[n],\ldots,w_{L-1}[n]]^{\T}$. 
Stacking all $N$ temporal measurements, the $L \times N$ OFDM received signal matrix is
\begin{equation}
\label{Neg1000}
\mbY_{_{\textrm{OFDM}}}=\mbA\mbX\mbP(\bnu)+\mbN ,   
\end{equation}
where $\mbA=\Diag{\ba}$, $\mbN=[\mbw[0],\ldots,\mbw[N-1]]^{\T}$ and the Doppler information of the target  is collected in  
\begin{equation}
\mbP(\bnu)=[\mbp(0,\bnu),\ldots,\mbp(N-1,\bnu)],
\end{equation} 
and\par\noindent\small
\begin{align}
\mbX&= \mbD \odot \mbH, \label{eq:11}\\
\mbD&=\Diag{\balpha_{_{0}}^{\T},\ldots,\balpha_{_{L-1}}^{\T}},\label{eq:12}\\
\balpha_{_{l}}&=[\alpha_{_{l1}},\ldots,\alpha_{_{lM}}]^{\T},\label{eq:13}\\
\mbH&=\Diag{\mbh_{_{0}}^{\T},\ldots,\mbh_{_{L-1}}^{\T}},\label{eq:14}\\
\mbh_{l}&=[h_{_{l1}},\ldots,h_{_{lM}}]^{\T},\label{eq:15}
\end{align}\normalsize
We assume the noise
is from complex zero-mean Gaussian distribution and correlated with a  positive-definite covariance $\bSigma$. The columns of $\mbN$ are assumed to be (independent and identically distributed) i.i.d. Then, OFDM measurements are  distributed as 
\begin{equation}
\mbY_{_{\textrm{OFDM}}}\sim \mathcal{CN}(\mbA\mbX\mbP(\bnu),\mbI_{N} \otimes \bSigma)
\end{equation}
where $\mbI_{N} \otimes \bSigma$ is the covariance of the temporally white noise.
Our goal is to design a waveform that maximizes the detection of a moving target located at a given range. 

\section{Target Detection}
\label{sec:3}


In order  to decide whether a 
target is present in a particular known range-cell, we perform binary hypothesis testing between $\mathcal{H}_0$ (target-free hypothesis) and $\mathcal{H}_1$ (target-present hypothesis), that is 
\begin{align}
\label{ragen2}
\mathcal{H}_0:  & \quad \mbY_{_{\textrm{OFDM}}}=\mbN,  \\
\mathcal{H}_1:  &  \quad \mbY_{_{\textrm{OFDM}}}=\mbA\mbX\mbP(\bnu)+\mbN .
\end{align}
The likelihood ratio is \cite{bickel2015mathematical,sen2010adaptive} 
\begin{equation}
\label{ragen1}
\mathcal{L}\left(\mbY_{\text{OFDM}};\bnu\right)=\frac{f_{\mathcal{H}_1}\left(\mbY_{\text{OFDM}} ; \bnu, \mbX, \bSigma_1\right)}{f_{\mathcal{H}_0}\left(\mbY_{\text{OFDM }}; \bSigma_0\right)} ,
\end{equation}\normalsize
where $f_{\mathcal{H}_0}$ and $f_{\mathcal{H}_1}$ are the likelihood functions under $\mathcal{H}_0$ and $\mathcal{H}_1$, respectively and $\bnu$ is the  Doppler frequency under test. Since the $\bSigma$ and target parameters are unknown, we employ a generalized likelihood ratio test (GLRT) by replacing the unknowns with their maximum likelihood estimates (MLEs) in the $\mathcal{L}\left(\mbY_{\text{OFDM}};\bnu\right)$ 
to obtain the GLRT for our detection problem \eqref{ragen2} as 
\par\noindent\small
\begin{equation}
\label{Neg100}
\mathcal{T}_{_{\text{GLR}}}=\frac{f_{\mathcal{H}_1}\left(\mbY_{\text{OFDM}} ; \bnu, \widehat{\mbX}, \widehat{\bSigma}_1\right)}{f_{\mathcal{H}_0}\left(\mbY_{\text{OFDM }}; \widehat{\bSigma}_0\right)} \quad \underset{\mathcal{H}_2}{\overset{\mathcal{H}_1}{\gtrless}} \gamma,
\end{equation}\normalsize   
where $\widehat{\bSigma}_0$ and $\widehat{\bSigma}_1$ are the MLEs of $\boldsymbol{\bSigma}$ under $\mathcal{H}_0$ and $\mathcal{H}_1, \widehat{\mbX}$ is the MLE of $\mbX$ under $\mathcal{H}_1$, and $\gamma$ is the detection threshold~\cite{bickel2015mathematical,sen2010adaptive}.
The MLEs of unknown covariance matrices are 
\begin{align}\label{Neg101}
\widehat{\bSigma}_1&=\frac{1}{N} \left(\mbY_{\text{OFDM}}-\mbA\mbX\mbP(\bnu)\right)^{\H}\left(\mbY_{\text{OFDM}}-\mbA\mbX\mbP(\bnu)\right),\\
\widehat{\bSigma}_0&=\frac{1}{N} \mbY_{\text{OFDM}}^{\H}\mbY_{\text{OFDM}}.
\end{align}\normalsize 
Therefore,
 In the Gaussian noise scenario, the GLR becomes
 \par\noindent\small
\begin{equation}\label{Neg102}
\mathcal{T}_{_{\text{GLR}}}=\frac{\operatorname{det}\left(\mbY_{\text{OFDM}}^{\H}\mbY_{\text{OFDM}}\right)}{\operatorname{det}\left(\left(\mbY_{\text{OFDM}}-\mbA\mbX\mbP(\bnu)\right)^{\H}\left(\mbY_{\text{OFDM}}-\mbA\mbX\mbP(\bnu)\right)\right)}.
\end{equation}\normalsize
It follows from the Wilk's theroem~\cite{kay2009fundamentals} that, as $N \rightarrow \infty$, the GLRT statistic under   $\mathcal{H}_0$ asymptotically garners a complex chi-square (central) distribution, i.e., 
\begin{equation}
N \ln \mathcal{T}_{_{\GLR}} \sim \mathbb{C} \chi_{_{r L}}^2,
\end{equation}
and under   $\mathcal{H}_1$,
\begin{equation}
N \ln \mathcal{T}_{_{\GLR}} \sim \mathbb{C} \chi_{_{r L}}^{^2}(\delta),
\end{equation}
where $r=\operatorname{rank}\left(\mbP(\bnu)\right)$, $rL$ is degrees of freedom, and $\delta$ is the noncentrality parameter of chi-square distribution obtained in \cite{sen2010adaptive}  as
\begin{equation}\label{Neg109}
\delta=\Tr{\bSigma^{-1}\mbA\mbX\mbP(\bnu)\mbP^{\H}(\bnu)\mbX^{\H}\mbA^{\H}}.
\end{equation}
We have the probability of false alarm ($P_{_{\textrm{FA}}}$) and probability of  detection ($P_{_{\textrm{D}}}$) as
\par\noindent\small
\begin{align}\label{eq:probab}
P_{_{\textrm{FA}}}=\textrm{Pr}&\left(\mathcal{T}_{_{\GLR}} > \gamma | \mathcal{H}_0\right)=Q_{_{rL}}(\gamma) \nonumber\\
P_{_{\textrm{D}}}=\textrm{Pr}&\left(\mathcal{T}_{_{\GLR}} > \gamma | \mathcal{H}_1 \right)=Q^{M}_{\frac{rL}{2}}(\sqrt{\delta},\gamma)=Q^{M}_{_{\frac{rL}{2}}}\left(\sqrt{\delta},Q_{_{rL}}^{-1}(P_{_{\textrm{FA}}})\right)
\end{align}\normalsize
where $1-Q_{_{rL}}(.)$ is the cumulative distribution function (CDF) of chi-squared distribution and $1-Q^{M}_{_{\frac{rL}{2}}}(.,.)$
is the Marcum Q-function accounting for the   CDF of non-central chi-squared distribution~\cite{kay2009fundamentals,shynk2012probability}. The Marcum Q-function
is strictly increasing in $\delta$. Thus, to maximize $P_{\textrm{D}}$, we maximize $\delta=f(\ba,\mbv)$ with respect to parameters such as the OFDM coefficients $\ba$ and vector of all IRS phase-shifts, $\mbv=\left[\mbv^{\top}_{1},\mbv^{\top}_{2}, \cdots ,\mbv^{\top}_{M}\right]^{\top}\in \mathbb{C}^{MN_{m}}$. 
We define the \emph{signal-to-noise ratio} (SNR) matrix  $\mbA\mbX\mbP(\bnu)\mbP^{\H}(\bnu)\mbX^{\H}\mbA^{\H}$ and its trace as in~\eqref{Neg109}, as the \emph{SNR metric}\cite{bickel2015mathematical,sen2010adaptive}.

\nocite{ardeshiri2021adaptive}
\section{Joint Waveform and IRS Phase-Shift Design}\label{sec:5}
As discussed earlier in Section~\ref{sec:3}, the probability of detection for a given probability of false alarm is a monotonically increasing function of the \emph{SNR}. Hence, the design problem may also be formulated to maximize the SNR with respect to the system parameters. The  joint  waveform and IRS  phase-shift design  problem is\par\noindent\small
\begin{align}\label{eq:main_opt}
&\underset{\mbv,\ba}
{\textrm{maximize}} \quad f(\ba,\mbv)\nonumber\\
&\textrm{subject to} \quad \ba^{H}\ba=1.
\end{align}\normalsize
We resort to a task-specific cyclic algorithm, wherein we cyclically optimize~\eqref{eq:main_opt} for $\ba$ and $\mbv$ \cite{tang2020polyphase,soltanalian2013joint}. 
To ensure the unimodularity constraint of IRS phase-shift, we need a projection-based optimization method.
\subsection{OFDM Waveform Design}
Problem~\eqref{eq:main_opt}  with respect to $\ba$ is recast as\par\noindent\small
\begin{align}\label{eq:opt1}
\mathcal{P}_{1}:\;\;\; & \underset{\ba}{\textrm{maximize}}\quad \ba^{\H}\left[\left(\mbX\mbP(\bnu)\mbP^{\H}(\bnu)\mbX^{\H}\right)^{\T}\odot\bSigma^{-1}\right] \ba\nonumber\\
&\textrm{subject to} \quad \ba^{H}\ba=1.
\end{align}\normalsize
The optimal OFDM coefficient $\ba$ is the eigenvector  corresponding to the dominant eigenvalue of $\left(\mbX\mbP(\bnu)\mbP^{\H}(\bnu)\mbX^{\H}\right)^{\T}\odot \bSigma^{-1}$, evaluated based on the \emph{power method}\cite{van1996matrix} at each iteration $s$ as follows:\par\noindent\small
\begin{equation}
\label{Neg_115} 
\ba^{(s+1)}=\frac{\left[\left(\mbX\mbP(\bnu)\mbP^{\H}(\bnu)\mbX^{\H}\right)^{\T}\odot \bSigma^{-1}\right]\ba^{(s)}}{\left\|\left [\left(\mbX\mbP(\bnu)\mbP^{\H}(\bnu)\mbX^{\H}\right)^{\T}\odot \bSigma^{-1}\right ]\ba^{(s)}\right\|_{2}},\; s \geq 0.
\end{equation}\normalsize
\subsection{IRS Beamforming Design}
To design the phase-shift parameters, we maximize the SNR metric with respect to $\mbv$: \par\noindent\small
\begin{align}\label{eq:opt2}
\mathcal{P}_{2}:\;\;\;\; \underset{\mbv}
{\textrm{maximize}}  &\quad f(\ba,\mbv).
\end{align}\normalsize
In what follows, we show that $\mathcal{P}_2$ can be written as a unimodular quadratic program (UQP). To tackle $\mathcal{P}_{2}$ with respect to $\mbv$, we adopt a computationally efficient procedure of the \emph{power-method-like iterations (PMLI)} algorithm~\cite{soltanalian2014designing}.
This method closely resembles the widely used power method for computing the dominant eigenvalue/vector pairs of matrices.

The UQP is defined as
\begin{equation}
\label{eq:UQP}
\underset{\mbs \in \Omega^{n}}{\textrm{maximize}} \quad\mbs^{\H} \mbG \mbs.
\end{equation}
The sequence of  unimodular vectors at the $t$-th PMLI iteration is
\begin{equation}
\label{eq:UQP_it}
\mbs^{(t+1)}=e^{\textrm{j}\arg{\mbG\mbs^{(t)}}},
\end{equation}
leads to a monotonically increasing objective value for the UQP, when $\mbG$ is a positive semidefinite matrix~\cite{soltanalian2014designing,hu2016locating}.

The following Lemma~\ref{lem:lem1} states the required transformations of~\eqref{eq:opt2} to facilitate the application of the PMLI approach.
\begin{lemma}\label{lem:lem1} Define $\mbC=\left[\begin{array}{c|c|c}\bUpsilon_1&\cdots& \bUpsilon_L\end{array}\right]$,$\bUpsilon_l=\sum^{LM}_{i=1} \mbC_i (\baleph_{l}^{\T}\mbE_i^{\T}\otimes \mbI)$, $\mbC_i=\mbe_i \otimes \mbI_{_L}$, $\mbe_i$ as  a $LM \times 1 $ vector whose $i$-th element is unity and remaining elements are zero, $\baleph_{l}$ as a $L \times L$ matrix with $[\baleph]_{_{ll}}=1$ and zero everywhere else, and $\mbE_i=\textrm{vec}^{-1}_{_{(M,L)}}\left(\mbe_i\right)$. 
Then, 
\begin{equation}
\vec{\mbH}=\mbC \mbh,
\end{equation}
\end{lemma}
\begin{proof}Following the  definitions in Lemma.~\ref{lem:lem1}, we have \par\noindent\small
\begin{align}
\vec{\mbH}&=\vec{\Diag{\mbh_1^{\T},\ldots,\mbh_L^{\T}}}=\vec{\sum_{l=1}^{L}\baleph_{l}   \otimes \mbh_l^{\T}}\nonumber\\
&=\sum_{l=1}^{L}\vec{\baleph_{l}  \otimes \mbh_l^{\T}}=\sum_{l=1}^{L} \sum^{LM}_{i=1} \mbC_i \left(\baleph_{l}  \otimes \mbh_l^{\T}\right) \mbe_i \nonumber\nonumber\\
&=\sum_{l=1}^{L}\sum^{LM}_{i=1} \mbC_i  \left(\baleph_{l}  \otimes \mbh_l^{\T}\right)\vec{\mbE_i}\nonumber\\
&=\sum_{l=1}^{L}\sum^{LM}_{i=1} \mbC_i\vec{\mbh_l^{\T}\mbE_i\baleph_{l}}=\sum_{l=1}^{L}\sum^{LM}_{i=1} \mbC_i (\baleph_{l}^{\T}\mbE_i^{\T}\otimes \mbI)\mbh_l\nonumber\\
&=\sum_{l=1}^{L} \bUpsilon_l \mbh_l=\mbC \mbh.
\end{align}\end{proof}\normalsize

Subsequently, we will use Lemma.~\ref{lem:lem1} to propose a quadratic form with respect to $\mbh$ for the SNR metric.
\begin{proposition}
Denote $\mbh=[\mbh_1^{\T},\ldots,\mbh_L^{\T}]^{\T}$, $\mbW=\mbC^{\H}\mbU^{\H}\mathcal{A}\mbU\mbC$ and $\mbU=\Diag{\vec{\mbA\mbD}}$. Then, the SNR metric becomes\par\noindent\small
\begin{equation}\label{eq:opt3}
f(\ba,\mbv)=\mbh^{\H}\mbW\mbh.
\end{equation}\normalsize
\end{proposition}
\begin{proof}
Assume $\mbB=\mbA\mbX$. We recast the objective in~\eqref{eq:opt2} as \par\noindent\small
\begin{align}\label{eq:opt_recast}
f(\ba,\mbv)
&= \Tr{\mbP^{\H}(\bnu)\mbB^{\H}\bSigma^{-1} \mbB \mbP(\bnu)}\nonumber\\
&=\Tr{\mbB^{\H}\bSigma^{-1}\mbB\mbP(\bnu)\mbP^{\H}(\bnu)}\nonumber\\
&=\vec{\mbB^{\ast}}^{\T} \vec{\bSigma^{-1}\mbB\mbP(\bnu)\mbP^{\H}(\bnu)}\nonumber\\
&=\vec{\mbB}^{\H} \left(\left(\mbP(\bnu)\mbP^{\H}(\bnu)\right)^{\T}\otimes \bSigma^{-1}\right)\vec{\mbB}\nonumber\\
&=\vec{\mbB}^{\H}\mathcal{A} \vec{\mbB},
\end{align}\normalsize
where $ \mathcal{A}=\left(\mbP(\bnu)\mbP^{\H}(\bnu)\right)^{\T}\otimes\bSigma^{-1}$. We have
$\vec{\mbB}=\vec{\mbA\mbX}=\vec{\Diag{\ba}\left(\mbD\odot \mbH\right)}=\vec{\mbH}\odot\vec{\mbA\mbD}=\mbC\mbh\odot\vec{\mbA\mbD}$ and using Lemma.~\ref{lem:lem1}, $\vec{\mbB}=\mbU\mbC\mbh$. Substituting this in~\eqref{eq:opt_recast} completes the proof. 
\end{proof}
We now reformulate the SNR metric as a quartic function in the optimization  parameter $\mbv$.  
\begin{proposition}~\label{prop:2}
The SNR metric is quartic in phase-shifts, i.e.\par\noindent\small
\begin{align}
\label{eq_neg10}
f(\ba,\mbv)&=\mbv^{\H}\mbQ_1(\mbv)^{\H}\mbW \mbQ_1(\mbv)\mbv\nonumber,\\
&=\mbv^{\H}\mbQ_2(\mbv)^{\H}\mbW \mbQ_2(\mbv)\mbv,
\end{align}\normalsize
\end{proposition}
where\par\noindent\small
\begin{align}\label{eq:4}
\mbQ_{1}(\mbv)&=\left[\begin{array}{c|c|c}(\mbS_1\mbQ(\mbv))^{\T}&\cdots& (\mbS_L\mbQ(\mbv))^{\T} \end{array}\right]^{\T},\nonumber\\
\mbQ_{2}(\mbv)&=\left[\begin{array}{c|c|c}(\mbS_1\mbQ^{\prime}(\mbv))^{\T}&\cdots&(\mbS_L\mbQ^{\prime}(\mbv))^{\T}\end{array} \right]^{\T},\nonumber\\
\mbQ(\mbv)&=\Diag{\mbv_1\otimes \mbI_{N_m},\ldots,\mbv_{_{M}}\otimes \mbI_{N_m}},\nonumber\\
\mbQ^{\prime}(\mbv)&= \Diag{  \mbI_{N_m}\otimes\mbv_1,\ldots,\mbI_{N_m}\otimes\mbv_{_M}},\nonumber\\
\mbS_{l}&=\Diag{\vec{\mbS_{l1}}^{\T},\ldots,\vec{\mbS_{lM}}^{\T}},\nonumber\\
\mbS_{lm}&=\left(\mbb(\theta_{ir,m},f_l)\odot \mbb(\theta_{ti,m},f_l)\right)\left(\mbb(\theta_{ri,m},f_l)\odot \mbb(\theta_{it,m},f_l)\right)^{\T}.
\end{align}\normalsize
\begin{proof} 
Given $\bPhi_m=\Diag{\mbv_m}$, it is  straightforward to verify  from~\eqref{eq:1} that  we  have  $h_{_{lm}}=\mbv_m^{\T}\mbS_{_{lm}} \mbv_m$ and \par\noindent\small
\begin{align}\label{eq:2}
\mbh_{l}&=[h_{_{l1}},\ldots,h_{_{lM}}]^{\T}=\left[\Tr{\mbS_{_{l1}}\mbv_{_1}\mbv_{_1}^{\T}},\ldots,\Tr{\mbS_{_{lM}}\mbv_{_M}\mbv_{_M}^{\T}}\right]^{\T},\nonumber\\
&= \Diag{\vec{\mbS_{_{l1}}}^{\T},\ldots,\vec{\mbS_{_{lM}}}^{\T}}\nonumber\\
&\qquad\times \left[\vec{\mbv_{_1}\mbv_{_1}^{\T}},\ldots,\vec{\mbv_{_M}\mbv_{_{M}}^{\T}}\right]^{\T}. 
\end{align}\normalsize
Applying the identity  \cite{van1996matrix}
\begin{equation}
\label{eq:vectorization}
\operatorname{vec}\left(\mbv_{m}\mbv^{\top}_{m}\right)=\left(\mbI_{N_m}\otimes\mbv_{m}\right)\mbv_{m}=\left(\mbv_{m}\otimes\mbI_{N_m}\right)\mbv_{m},  
\end{equation}
to~\eqref{eq:2} produces
\begin{align}\label{eq:3}
\mbh_{_{l}}=\mbS_{_{l}}  \mbQ(\mbv)\mbv=\mbS_{_{l}}\mbQ^{\prime}(\mbv)\mbv,
\end{align}
where $\mbS_{_l}$, $\mbQ(\mbv)$ and $\mbQ^{\prime}(\mbv)$ are given in~\eqref{eq:4}. 
Concatenating  the vectors in~\eqref{eq:3}, we obtain 
\begin{equation}
\mbh=\mbQ_{1}(\mbv)\mbv=\mbQ_{2}(\mbv)\mbv,   
\end{equation}
where $\mbQ_{1}(\mbv)$ and $\mbQ_{2}(\mbv)$ are  given in~\eqref{eq:4}.
\end{proof}
\vspace{-6pt}
\subsection{Proposed Algorithm}
We cyclically tackle the SNR maximization via its bi-quadratic transformation  with respect to auxiliary variables $\mbv_{(1)}$ and $\mbv_{(2)}$. In the sequel, $\mbv_{(1)},\mbv_{(2)}\in \mathbb{C}^{MN_m}$ are vectors produced by symmetrization, representing the collection of phase-shifts of all IRS platforms, whereas in the previous parts $\mbv_{1},\mbv_{2}\in \mathbb{C}^{N_m}$ were the vector of phase-shift of IRS$_1$ and IRS$_2$. 
In the following proposition, we  recast the SNR metric as a  bi-quadratic function of  $\mbv_{(1)}$ and $\mbv_{(2)}$. 
\begin{proposition}
\label{theo:1} The function \par\noindent\small
\begin{align}\label{eq:g}
g(\mbv_{_{(1)}},\mbv_{_{(2)}})&=\mbv_{_{(1)}}^{\H}\mbE(\mbv_{_{(2)}})\mbv_{_{(1)}},\nonumber\\
&=\mbv_{_{(2)}}^{\H}\mbE(\mbv_{_{_{(1)}}})\mbv_{_{(2)}},
\end{align}\normalsize
where \par\noindent\small
\begin{align}\label{eq:E(v)}
\mbE(\mbv)&=\frac{\mbG_1(\mbv)+\mbG_2(\mbv)}{2},\nonumber\\
\mbG_1(\mbv)&=\mbQ^{\H}_{1}(\mbv)\mbW\mbQ_{1}(\mbv),\nonumber\\
\mbG_2(\mbv)&=\mbQ^{\H}_{2}(\mbv)\mbW\mbQ_{2}(\mbv),
\end{align}\normalsize
is a bi-quadratic transformation of SNR metric in (\ref{eq_neg10}). 
\end{proposition}
\begin{proof}
To show symmetry,  from~\eqref{eq:4} and~\eqref{eq:vectorization}, we observe that $\mbQ_1(\mbv_{_{(i)}})\mbv_{_{(k)}}=\mbQ_2(\mbv_{_{(k)}})\mbv_{_{(i)}}$ for $ i\neq k \in \left\{1,2\right\}$. Therefore,\par\noindent\small
\begin{align}
&g \left (\mbv_{_{(1)}},\mbv_{_{(2)}} \right)= \mbv_{_{(1)}}^{\H}\mbE\left (\mbv_{_{(2)}}\right )\mbv_{_{(1)}}\nonumber\\
&=\mbv_{_{(1)}}^{\H}\frac{\mbG_1 \left(\mbv_{_{(2)}}\right)+\mbG_2\left(\mbv_{_{(2)}}\right)}{2}\mbv_{_{(1)}}\nonumber\\
&=\mbv_{_{(1)}}^{\H}\frac{\mbQ^{\H}_{1}\left(\mbv_{_{(2)}}\right)\mbW\mbQ_{1}\left(\mbv_{_{(2)}}\right)+\mbQ^{\H}_{2}\left(\mbv_{_{(2)}}\right)\mbW\mbQ_{2}\left(\mbv_{_{(2)}}\right)}{2}\mbv_{_{(1)}}\nonumber\\
&=\mbv_{_{(2)}}^{\H}\frac{\mbQ^{\H}_{2}\left(\mbv_{_{(1)}}\right)\mbW\mbQ_{2}\left(\mbv_{_{(1)}}\right)+\mbQ^{\H}_{1}\left(\mbv_{_{(1)}}\right)\mbW\mbQ_{1}\left(\mbv_{_{(1)}}\right)}{2}\mbv_{_{(2)}}\nonumber\\
&=\mbv_{_{(2)}}^{\H}\frac{\mbG_2\left(\mbv_{_{(1)}}\right)+\mbG_1\left(\mbv_{_{(1)}}\right)}{2}\mbv_{_{(2)}}\nonumber\\
&=\mbv_{_{(2)}}^{\H}\mbE\left(\mbv_{_{(1)}}\right)\mbv_{_{(2)}}=g\left(\mbv_{_{(2)}},\mbv_{_{(1)}}\right).
\end{align}\normalsize
By Substituting $\mbv_{_{(1)}}=\mbv_{_{(2)}}=\mbv$ and ~\eqref{eq:E(v)} in~\eqref{eq:g}, and comparing it with~\eqref{eq_neg10},  one can verify that $f(\ba,\mbv)=g(\mbv,\mbv)$. 
\end{proof}
Since $f(\ba,\mbv)=g(\mbv,\mbv)$, we propose to maximize $f(\ba,\mbv)$ by  alternately fixing  $\mbv_{(1)}$ or $\mbv_{(2)}$ and maximizing $g\left(\mbv_{_{(1)}},\mbv_{_{(2)}}\right)$ with respect to the other variable while enforcing  $\mbv_{_{(1)}}=\mbv_{_{(2)}}$  as a constraint. From Proposition~\ref{theo:1}, fixing either $\mbv_{_{(1)}}$ or $\mbv_{_{(2)}}$ and maximizing $g\left(\mbv_{_{(1)}},\mbv_{_{(2)}}\right)$  with respect to the other variable requires solving the following \par\noindent\small
\begin{align}
\label{eq:UQP1}
\underset{\mbv_{(k)}\in\Omega^{MN_{m}}}
{\textrm{maximize}} \quad \mbv_{(k)}^{\mathrm{H}} \mbE \left(\mbv_{(i)}\right) \mbv_{(k)}, \quad i\neq k \in \left\{1,2\right\},
\end{align}\normalsize
\begin{remark}\label{rmk:1} 
In UQP, the diagonal loading technique is  used to ensure the positive semidefiniteness of the matrix, without changing the optimal solution~\cite{esmaeilbeig2022joint}.
In~\eqref{eq:UQP1},  the diagonal loading  as
 $\widetilde{\mbE}(\mbv) \leftarrow \lambda_{m} \mbI - \mbE(\mbv)$, with  
$\lambda_{_{m}}$  being the maximum eigenvalue of $\mbE(\mbv)$,  results is an equivalent problem. 

Note that diagonal loading has no effect on the solution of (\ref{eq:UQP1}) because $\mbv^{\mathrm{H}} \widetilde{\mbE}(\mbv) \mbv=\lambda_{_{m}}MN_m- \mbv^{\mathrm{H}} \mbE(\mbv) \mbv$. The  equivalent problem to~\eqref{eq:UQP1} is\par\noindent\small
\begin{align}\label{eq:UQP2}
\underset{\mbv_{(k)}\in\Omega^{MN_{m}}}
{\textrm{minimize}} \quad \mbv_{(k)}^{\mathrm{H}} \widetilde{\mbE}\left(\mbv_{(i)}\right) \mbv_{(k)}, \quad i\neq k \in \left\{1,2\right\}.
\end{align}\normalsize
\end{remark}

The following theorem demonstrates that the IRS beamforming design problem  $\mathcal{P}_2$  is equivalent to a unimodular bi-quadratic programming (UBQP)  that we solve using the PMLI approach in~\eqref{eq:UQP_it}. 
\begin{theorem}
\label{theo:2} 
The SNR maximization problem $\mathcal{P}_2$ with respect to phase-shifts  is equivalent to the following UBQP\par\noindent\small 
\begin{align}\label{eq:8}
\underset{\mbv_{_{(k)}}\in\Omega^{MN_{m}}}
{\textrm{maximize}}& \quad
\begin{bmatrix} \mbv_{(k)}\\1\end{bmatrix}^{\H}
 \underbrace{\begin{bmatrix} \widehat\lambda_{m}\mbI-\widetilde{\mbE}(\mbv_{_{(i)}})& \eta\mbv_{_{(i)}} \\ \eta \mbv_{_{(i)}}^{\mathrm{H}} & \widehat\lambda_{m}-2\eta M N_{m} \end{bmatrix}}_{=\widehat{\mbE}\left(\mbv_{_{(i)}}\right)}
 \begin{bmatrix} \mbv_{(k)}\\1\end{bmatrix},
\end{align}\normalsize
where $i\neq k \in \left\{1,2\right\}$, and $\widehat\lambda_{m}$ is the maximum eigenvalue of $\mathcal{E}(\mbv_{_{(i)}})$ as defined in~\eqref{eq:9}. 
\end{theorem}
\begin{proof} From Proposition~\ref{theo:1} and Remark~\ref{rmk:1}, we  know that  $\mathcal{P}_2$ is   equivalent to  the following problems, 
\begin{align}\label{eq:5}
&\underset{\mbv_{(k)}\in\Omega^{MN_{m}}}
{\textrm{minimize}} \quad \mbv_{(k)}^{\mathrm{H}} \widetilde{\mbE}(\mbv_{(i)}) \mbv_{(k)}, \quad i\neq k \in \left\{1,2\right\},\nonumber\\
&\textrm{subject to} \quad \mbv_{_{(i)}}=\mbv_{_{(k)}}.
\end{align}
We add the $\ell_2$-norm penalty term between $\mbv_{_{(1)}}$ and $\mbv_{_{(2)}}$ as a \emph{penalty} function to~\eqref{eq:5}, which yields\par\noindent\small
\begin{equation}\label{eq:6}
\underset{\mbv_{_{(k)}}\in\Omega^{MN_{m}}}
{\textrm{minimize}} \quad \mbv_{_{(k)}}^{\H} \widetilde{\mbE}(\mbv_{_{(i)}}) \mbv_{_{(k)}} + \eta \|\mbv_{_{(i)}} - \mbv_{_{(k)}}\|_2^2 , \quad i\neq k \in \left\{1,2\right\},
\end{equation}\normalsize
where $\eta$ is a  Lagrangian multiplier. The regularizer as well as the main objective are quadratic in  $\mbv_{_{(1)}}$ and $\mbv_{_{(2)}}$. Consequently, we recast the objective of~\eqref{eq:6} as\par\noindent\small
\begin{align}\label{eq:7}
 &\mbv_{_{(k)}}^{\H} \widetilde{\mbE}(\mbv_{_{(i)}}) \mbv_{_{(k)}} + \eta \|\mbv_{_{(i)}} - \mbv_{_{(k)}}\|_2^2, \nonumber \\
 &=\mbv_{_{(k)}}^{\H} \widetilde{\mbE}(\mbv_{_{(i)}}) \mbv_{_{(k)}} -2\eta \Re{\mbv_{_{(k)}}^{\H}\mbv_{_{(i)}}}+2\eta MN_m, \nonumber\\
 &=\begin{bmatrix}\mbv_{_{(k)}} \\ 1\end{bmatrix}^{\H}
 \begin{bmatrix} \widetilde{\mbE}(\mbv_{_{(i)}}) & \eta \mbv_{_{(i)}} \\ \eta \mbv_{_{(i)}}^{\mathrm{H}} & 2\eta M N_{m} \end{bmatrix}
 \begin{bmatrix}\mbv_{_{(k)}} \\ 1\end{bmatrix},
\end{align}\normalsize
where in the first equality, we used  $\|\mbv_{_{(k)}}\|_2^2=MN_m$, due to unimodularity of  $\mbv_{_{(k)}}$. Substituting~\eqref{eq:7} in~\eqref{eq:6} yields
\begin{align}\label{eq:9}
\underset{\mbv_{_{(k)}}\in\Omega^{MN_{m}}}
{\textrm{minimize}} &\quad \begin{bmatrix} \mbv_{(k)}\\1\end{bmatrix}^{\H}
 \underbrace{\begin{bmatrix} \widetilde{\mbE}(\mbv_{_{(i)}}) & \eta \mbv_{_{(i)}} \\ \eta \mbv_{_{(i)}}^{\mathrm{H}} & 2\eta M N_{m} \end{bmatrix}}_{=\mathcal{E}(\mbv_{_{(i)}})}
 \begin{bmatrix} \mbv_{(k)}\\1\end{bmatrix}, \nonumber\\&\quad i\neq k \in \left\{1,2\right\}.    
\end{align}
We use  diagonal loading as introduced in  Remark~\ref{rmk:1} to obtain~\eqref{eq:8}.
\end{proof}

Based on Theorem~\ref{theo:2}, the  IRS phase assignment problem can be formulated as a UBQP   and therefore it may be tackled in an alternating  manner over $\bar{\mbv}_{_{(1)}}=[\mbv^{\T}_{_{(1)}} 1]^{\T}$ and $\bar{\mbv}_{_{(2)}}=[\mbv^{\T}_{_{(2)}} 1]^{\T}$, by  the PMLI iterations 
in \eqref{eq:UQP_it}.
The PMLI has been shown to be convergent in terms of both the optimization objective and variable\cite{soltanalian2013joint,hu2016locating}. Algorithm~\ref{algorithm_1} summarizes the steps for joint waveform and IRS phase-shift design.
\begin{algorithm}[H]
\caption{Joint IRS phase-shift and OFDM waveform design}
    \label{algorithm_1}
    \begin{algorithmic}[1]
    \Statex \textbf{Input:} 
    Initialization values $\mbv^{(0)}_{1}$,$\mbv^{(0)}_{2}$, $\ba^{(0)}$, the Lagrangian multiplier $\eta$, total number of iterations $\Gamma_{1}$ ($\Gamma_{2}$) for the problem $\mathcal{P}_1$ ($\mathcal{P}_2$). 
    \vspace{1.2mm}
    \Statex \textbf{Output:} Optimized phase-shifts $\mbv^{\star}$ and  OFDM signal coefficients $\ba^{\star}$.
    \vspace{1.2mm}
    \For{$s=0:\Gamma_{1}-1$} \hspace{4mm}$\triangleright$ $\mbv^{(t)}_{(1)}$ and $\mbv^{(t)}_{(2)}$ are the solutions at the $t$-th iteration. 
     \vspace{1.2mm}
    \For{$t=0:\Gamma_{2}-1$} 
    \vspace{1.2mm}\small
     \State  $\mbv^{(t+1)}_{(1)}\gets e^{\j \arg{\begin{bmatrix} \mbI_{_{M N_{m}}} \mathbf{0}_{_{M N_{m}}}\end{bmatrix}\left(\widehat{\lambda}_{m}\mbI-\mathcal{E}(\mbv^{(t+1)}_{(2)},\ba^{(s)})\right)\bar{\mbv}^{(t)}_{(1)}}}$.
     \vspace{1.2mm}
    \State  $\mbv^{(t+1)}_{(2)}\gets e^{\j \arg{\begin{bmatrix} \mbI_{_{M N_{m}}} \mathbf{0}_{_{M N_{m}}}\end{bmatrix}\left(\widehat{\lambda}_{m}\mbI-\mathcal{E}(\mbv^{(t)}_{(1)},\ba^{(s)})\right)\bar{\mbv}^{(t)}_{(2)}}}$,

     \vspace{1.2mm}\normalsize
    \EndFor
    \State $\mbv^{(s)}\gets \mbv^{(\Gamma_{2})}_{(1)}$ or $\mbv^{(\Gamma_{2})}_{(2)}$.
    \vspace{1.2mm}
    \State $\mbh^{(s)} \gets   \mbQ_{1}(\mbv^{(s)})\mbv^{(s)} $.
    \vspace{1.2mm}
    \State Update $\mbX^{(s)}$ according to~\eqref{eq:11}-\eqref{eq:14}.
    \vspace{1.2mm}
    \State  Update $\ba^{(s)}$ according to~\eqref{Neg_115}. 
    \EndFor
    \vspace{1.2mm}
    \State \Return $\left\{ \ba^{\star},\mbv^{\star}\right\} \gets \left \{ \ba^{(\Gamma_1)},\mbv^{(\Gamma_1)}\right\}$.
    \end{algorithmic}
\end{algorithm}\normalsize
\section{Numerical analysis}\label{sec:6}
We performed numerical experiments to analyze the performance of  Algorithm~\ref{algorithm_1} through the RoC of the proposed detector. The radar located at $\uprho_r=[0,0]^{\T}$ was set to transmit $N=50$ pulses  with pulse-width $T=50 $ ns, carrier frequency $f_c=1$ GHz, bandwidth $B=100$ MHz, and PRI $T_{_{\PRI}}=20 \mu$s. The OFDM signal had $L=4$ subcarriers with spacing $\Delta_f=1/T=20$ MHz. The IRS$_1$ was located at $\uprho_i^{(1)}=[0.1,0.1]^{\T}$km and IRS$_2$ is located at $\uprho_i^{(2)}=[-0.1,0.1]^{\T}$km. Each IRS comprises $N_m=8$ elements arranged as a ULA. The target was located at $\uprho_{_t}=[0,5]^{\T}$ km moving with  $\upnu=[10,10]^{T}$m/s. The complex   target  reflectivity  coefficients  $\alpha_{_{lm}}$ corresponding to a Swerling-0 target model were drawn from $\mathcal{CN}(0,1)$. 

We compared multi-IRS-aided OFDM radar with a multipath OFDM radar~\cite{sen2010adaptive}. Fig.~\ref{fig_N} illustrates the RoC obtained after $10^3$ Monte-Carlo trials for fixed probability of false alarm $P_{_{\textrm{FA}}}$. In each trial, the  threshold $\gamma$  is set according to \eqref{eq:probab} for the desired $P_{_{\textrm{FA}}}$. We observe that deploying $M=1$ IRS platforms with phase-shifts obtained by Algorithm 1 improves $P_{_\textrm{D}}$ 
over the non-IRS (LoS) OFDM radar. Moreover, for $M=2$, IRS-aided  outperforms multipath  OFDM radar proposed in~\cite{sen2010adaptive}  with specular reflection $\{h_{_{lm}}\}=1$ in the exactly identical two paths between the target and radar.   Deploying multiple IRS platforms  in comparison with  single  IRS and non-IRS scenarios provides additional degrees of freedom (DoFs) and  improves  performance. Also, multi-IRS-aided radar  outperforms single IRS because an optimal deployment of more IRSs provides more NLoS paths and, hence, enhanced detection of NLoS targets especially those that may not be accessible via only one IRS.
\begin{figure}[t]
\centering
	\includegraphics[width=1.0\columnwidth]{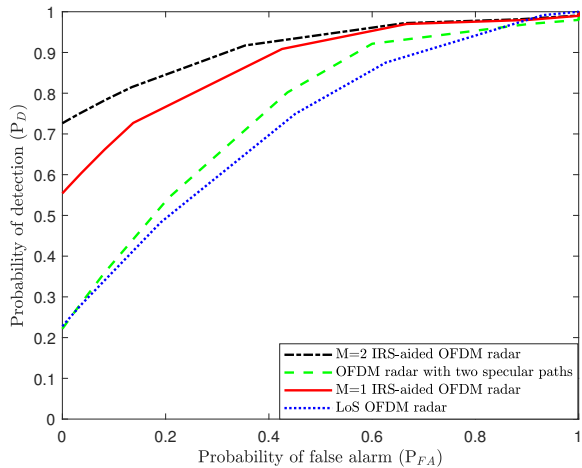}
	\caption{RoC of   detection for LoS OFDM radar (single path), OFDM radar  with 2 specular paths, $M=1$ IRS-aided OFDM radar,   $M=2$ IRS-aided OFDM radar. 
 }
\label{fig_N}
\end{figure}
\vspace{-6pt}\section{Summary}\label{sec:7}
We investigated  the  moving target detection problem  using a multi-IRS-aided OFDM radar. The IRS phase-shifts along with the  OFDM transmit signal coefficients were designed by taking advantage of an alternating optimization of the non-centrality parameter of the GLRT. We  showed that maximizing the non-centrality parameter improved the probability of detection. By means of numerical investigations, we demonstrated that our proposed method enhances the $P_{_{\textrm{D}}}$ over the non-IRS OFDM  radar systems.
\bibliographystyle{IEEEtran}
\bibliography{ref}
\end{document}

%% file: symbols.tex
\usepackage{amsmath}
\usepackage{amsthm}
\usepackage{dsfont}
\usepackage{thmtools}
\usepackage{amssymb}
\usepackage{xcolor}

\def\ba{\boldsymbol{a}}

\def\balpha{\boldsymbol{\alpha}}

\def\bGamma{\boldsymbol{\Gamma}}

\def\btheta{\boldsymbol{\theta}}

\def\bnu{\boldsymbol{\nu}}

\def\baleph{\boldsymbol{\aleph}}

\def\bSigma{\boldsymbol{\Sigma}}

\def\bUpsilon{\boldsymbol{\Upsilon}}

\def\bPhi{\boldsymbol{\Phi}}

\def\mbb{\mathbf{b}}

\def\mbe{\mathbf{e}}

\def\mbh{\mathbf{h}}

\def\mbp{\mathbf{p}}

\def\mbr{\mathbf{r}}
\def\mbs{\mathbf{s}}

\def\mbv{\mathbf{v}}
\def\mbw{\mathbf{w}}
\def\mbx{\mathbf{x}}
\def\mby{\mathbf{y}}
\def\mbz{\mathbf{z}}

\def\mbA{\mathbf{A}}
\def\mbB{\mathbf{B}}
\def\mbC{\mathbf{C}}
\def\mbD{\mathbf{D}}
\def\mbE{\mathbf{E}}

\def\mbG{\mathbf{G}}
\def\mbH{\mathbf{H}}
\def\mbI{\mathbf{I}}

\def\mbK{\mathbf{K}}

\def\mbN{\mathbf{N}}
\def\mbO{\mathbf{O}}
\def\mbP{\mathbf{P}}
\def\mbQ{\mathbf{Q}}
\def\mbR{\mathbf{R}}
\def\mbS{\mathbf{S}}

\def\mbU{\mathbf{U}}

\def\mbW{\mathbf{W}}
\def\mbX{\mathbf{X}}
\def\mbY{\mathbf{Y}}

\def\bzero{\boldsymbol{0}}

\def\Tr#1{\mathrm{Tr}\left(#1\right)}
\def\vec#1{\mathrm{vec}\left(#1\right)}

\def\Diag#1{\mathrm{Diag}\left(#1\right)}

\def\Re#1{\operatorname{Re}\left(#1\right)}
\def\Im#1{\operatorname{Im}\left(#1\right)}
\def\arg#1{\operatorname{arg}\left(#1\right)}

\newtheorem{theorem}{Theorem}
\newtheorem{proposition}{Proposition}
\newtheorem{lemma}{Lemma}

\newtheorem{remark}{Remark}

\def\PRI{{\textrm{PRI}}}
\def\T{\top}
\def\H{\mathrm{H}}
\def\j{\mathrm{j}}

\def\PRI{{\textrm{PRI}}}
\def\GLR{{\textrm{GLR}}}